\documentclass[onecolumn]{IEEEtran}
\onecolumn
\usepackage{amsfonts}
\usepackage{mathrsfs}
\usepackage{amsmath}
\usepackage{amssymb}
\usepackage{graphicx}
\usepackage{epic}
\usepackage{epstopdf}
\usepackage{color}
\usepackage{cite}
\usepackage{mathtools}
\usepackage{enumerate}
\usepackage{lineno}
\usepackage{url}
\usepackage{float}
\usepackage{color}
\usepackage{bm}
 \usepackage{booktabs}
\usepackage{tikz}
\usepackage[justification=centering]{caption}

\newtheorem{theorem}{Theorem}[section]
\newtheorem{corollary}[theorem]{Corollary}

\newtheorem{example}[theorem]{Example}

\newenvironment{proof}{\noindent {\bf Proof.}}{\rule{3mm}{3mm}\par\medskip}

\begin{document}
	
	\title{Network Function Computation With Different Secure Conditions}
	
	\author{Min Xu,  Gennian Ge and Minqian Liu
		
		\thanks{M. Xu ({\tt minxu0716@qq.com}) and M. Liu ({\tt mqliu@nankai.edu.cn}) are with the School of Statistics and Data Science, LPMC \& KLMDASR,
            Nankai University, Tianjin 300071, China. The research of M. Liu was supported by the National Natural Science Foundation of China under Grant No. 12131001, and National Ten Thousand Talents Program.}
		\thanks{G. Ge ({\tt gnge@zju.edu.cn}) is with the School of Mathematical Sciences, Capital Normal University,
			Beijing 100048, China. The research of G. Ge was supported by the National Key Research and Development Program of China under Grant Nos. 2020YFA0712100  and  2018YFA0704703, National Natural Science Foundation of China under Grant No. 11971325, and Beijing Scholars Program.

	}}
	\date{}
	\maketitle
    \begin{abstract}
		In this paper, we investigate function computation problems under different secure conditions over a network with multiple source nodes and a single sink node which desires a function of all source messages without error. A wiretapper has access to some edges of the network. Based on different practical requirements, we consider two secure conditions named as secure and user secure respectively. The main parameter concerned here is the computing rate, which is the average times of the target function that can be computed securely or user securely without error for one use of the network. In the secure case, a new upper bound which is tighter than the previous one is provided for arithmetic sum functions and arbitrary networks. Moreover, we show that the improved upper bound is strictly tight for tree-like networks. In the user secure case, we give a sufficient and necessary condition for the existence of user secure network codes and obtain an upper bound for the computation capacity.
	\end{abstract}

    \section{Introduction}
In this paper, we study the problem of function computation over a communication network with two different secure conditions.
The general setup of secure network function computation can be modeled as follows.
In a communication network characterized by a directed acyclic graph, there is a sink node which wants to compute a function $f$ about the messages generated at source nodes independently and uniformly.
Meanwhile, there is a wiretapper who has access to some edge subsets $W$ of the graph, we use $\mathcal{W}$ to denote all possible wiretap sets.
Each link of the network has unit transmission capacity and each node has unlimited computation capacity.
We hope the sink node can compute the function without error and the wiretapper can not get any information about the source message or the result of the function.

In last few years, network function computation has attracted considerable attention due to its application in sensor networks \cite{GK2005} and data processing \cite{DPK2016}, etc.
The security problem is necessary to be studied when the network function computation is applied to practice.

\subsection{Related Works}
The main parameter concerned in this paper is the computation capacity for the sink node to reliably compute the target function over the network under the secure conditions.
It's obvious that when $\mathcal{W}=\emptyset$, the secure network function computation problem reduces to non-secure case, which has been studied in recent years \cite{RD2012,AF2013,AFK2011,HTYG2018,GYYL2019}.
When the target function is linear, given the network, the capacity of the network is completely determined and the construction of achievable network codes is given \cite{AF2013}.
However, for general target functions, there are some cut-set type upper bounds of the computation capacity \cite{AFK2011,HTYG2018,GYYL2019}, but these upper bounds are not always tight.
Thus, we will focus on linear target function over finite field in this paper.

There have been extensive studies and applications of the security problem in network coding theory.
In 2002, Cai and Yeung \cite{CY2002} first investigated the existence of secure network codes.
In network coding problem, there is a single source node and multiple sink nodes, each wants a copy of the source message.
The secure condition means the wiretapper cannot get any information about the source message.
And some other secure conditions have also been studied \cite{BN2005,HY2008}.
For secure network function computing, limited work has been done.
Guang et al. \cite{GBY2021} developed an upper bound of the capacity of secure function computing for linear target functions in 2021.
And Li et al. \cite{LL2022} considered another secure condition which they called secret sum.

\subsection{Contributions and Organization}
Our main results are as follows.
On one hand, we show that Guang's upper bound in \cite{GBY2021} is not tight by an example.
From this example, we establish a new upper bound of the capacity under the secure condition.
Moreover, we show that for a special class of networks, our new upper bound is tight.
On the other hand, we define a new secure condition based on practical requirements, and investigate the same problem under the new secure condition.
A sufficient and necessary condition for the existence of secure network codes is addressed and an upper bound of the computing capacity is obtained.

This paper is organized as follows.
Section \ref{s2} introduces network function computation problems under two different secure conditions and reviews some existing bounds in related research works.
In Sections \ref{s3} and \ref{s4}, we focus on the capacity under the secure condition in \cite{GBY2021}.
We establish our new upper bound from an example and prove it in Section~\ref{s3}.
In Section~\ref{s4}, we derive the condition of the linear transformation which turns a non-secure network code into a secure one, and show that for a special class of networks, our new upper bound is tight.
For the new security, we address a sufficient and necessary condition for the existence of secure network codes and establish an upper bound by information-theoretical approach in Section \ref{s5}.
Finally, we give a conclusion of all results in Section \ref{s6}.

\section{Model and Preliminaries}\label{s2}
\subsection{Model and Secure Conditions}
Let $\mathcal{N}=(V,E)$ be a directed acyclic graph where $V$ is the node set and $E$ is the edge set.
For any edge $e\in E$, we use $tail(e)$ and $head(e)$ to denote the tail and head of $e$. For any vertex $v\in V$, let $In(v)=\{e\in E:head(e)=v\}$ and $Out(v)=\{e\in E:tail(e)=v\}$, respectively.
An edge sequence $\{e_1,e_2,\cdots,e_n\}$ is a $path$ from node $u$ to node $v$ if $tail(e_1)=u,head(e_n)=v$ and $tail(e_{i+1})=head(e_i)$ for $i=1,2,\cdots,n-1$.
For two nodes $u,v\in V$, a $cut$ of them is an edge set $C$ such that every path from $u$ to $v$ contains some edges in $C$.
If $U\subset V$ is a set of nodes, the cut of $U$ and $v$ is the edge set separating any vertex $u\in U$ from $v$.
The minimum size of a cut of $u$ and $v$ is denoted as $mincut(u,v)$.
There are two types of special nodes in the graph. Source nodes are the nodes which have no incoming edges and sink nodes are nodes which have no outgoing edges.
In this paper, we focus on multi-source single terminal network, denote the source nodes as $S=\{\sigma_1,\sigma_2,\cdots,\sigma_s\}$ and the sink node as $\rho$.

Generally, a target function $f: \mathcal{A}^s\mapsto\mathcal{O}$ is needed to be computed at sink node $\rho$,
where $\mathcal{A},\mathcal{O}$ are finite fields. We assume that each source node in the network $\mathcal{N}$ generates a sequence of random symbols over $\mathcal{A}$ independently and uniformly, and each edge of the network transmits a symbol of $\mathcal{B}$ once a time, where $\mathcal{B}$ is also a finite field.
Let the message generated at source node $\sigma_i$ be a random variable $M_i$, where $M_1,M_2,\cdots,M_s$ are mutually independent.
Computing the function $l$ times implies source $\sigma_i$ generates a random vector $\mathbf{M}_i=(M_{i1},\cdots,M_{il})$, and denote all the source messages by $\mathbf{M}_S=(\mathbf{M}_1,\cdots,\mathbf{M}_s)$.

In secure network function computation, there is a wiretapper who wants information about messages in the network.
Let $\mathcal{W}$ be all possible wiretap sets, and the wiretapper can access exactly one edge subset $W\in\mathcal{W}$.
$\mathcal{W}$ is known by the source nodes. The sink node need to compute target function
\begin{center}
  $f(\mathbf{M}_S)=(f(M_{1j},M_{2j},\cdots,M_{sj}):j=1,2,\cdots,l)$
\end{center}
without error while the wiretapper cannot get any useful information.
In network coding case, different secure conditions such as secure \cite{CC2011,CY2002}, weakly secure \cite{BN2005}, strongly secure \cite{HY2008} are considered.
In network function computation case, Guang et al. \cite{GBY2021} focused on protecting all source messages from the wiretapper.
Sometimes, in sensor network, the original messages generated at sources are not important, while the computation result is.
Li et al. considered protecting the arithmetic sum of all source messages from the wiretapper\cite{LL2022}.
They assumed each edge has unlimited transmission capacity, i.e. each edge can send unlimited symbols of $\mathcal{B}$ once a time.
We are concerned in this paper with the following two secure conditions.
\begin{equation}\label{secure}
  (secure)\ \ \ \ \ \ \ I(\mathbf{M}_S;\mathbf{Y}_W)=0,
\end{equation}
\begin{equation}\label{usecure}
  (user\ secure)\ \ I(f(\mathbf{M}_S);\mathbf{Y}_W)=0,
\end{equation}
where $\mathbf{Y}_W$ is the message transmitted by the edges in $W$ and observed by the wiretapper. In network coding, the message sent by each node is a function of all the messages the node has. Therefore, random keys are needed under the secure condition. Otherwise, any edge could leakage the information about $\mathbf{M}_S$.
But for the user secure case, random keys are not necessary.

An $(l,n)$ network code for the secure (or user secure) problem $(\mathcal{N},f,\mathcal{W})$ is defined as follows.
First, we let $\mathbf{m}_i\in\mathcal{A}^l$ and $\mathbf{k}_i\in\mathcal{K}_i$ be the outputs of the source message $\mathbf{M}_i$ and the key $\mathbf{K}_i$ respectively, $\mathcal{K}_i=\emptyset$ if random keys are not necessary.
An $(l,n)$ network code consists of a local encoding function for each edge $e$ which is defined as follows,
\begin{equation}\label{2}
  \widehat{\theta}_e:
  \begin{cases}
    \mathcal{A}^l\times\mathcal{K}_i\mapsto\mathcal{B}^n, & \mbox{if $tail(e)= \sigma_i$ for some $i$}; \\
    \prod\limits_{d\in In(tail(e))}\mathcal{B}^n\mapsto\mathcal{B}^n, & \mbox{otherwise},
  \end{cases}
\end{equation}
and a decoding function $\widehat{\varphi}:\prod_{In(\rho)}\mathcal{B}^n\mapsto\mathcal{O}^l$.
By the definition of local encoding function, the message transmitted by edge $e$ denoted as $U_e$ is a function of all source messages and random keys, i.e. for each $e\in E$,
\begin{equation}\label{3}
  \widehat{g}_e(\mathbf{m}_S,\mathbf{k}_S)=
  \begin{cases}
    \widehat{\theta}_e(\mathbf{m}_i,\mathbf{k}_i), & \mbox{if $tail(e)= \sigma_i$ for some $i$}; \\
    \widehat{\theta}_e(\widehat{g}_{In(tail(e))}(\mathbf{m}_S,\mathbf{k}_S)), & \mbox{otherwise}.
  \end{cases}
\end{equation}
We call $\widehat{g}_e$ as global encoding function of edge $e$.
An $(l,n)$ network code works for secure or user secure function computation tasks over a network if the sink node $\rho$ can compute target function without error, i.e. for all $\mathbf{m}_S\in\mathcal{A}^{s\cdot l}$ and $\mathbf{k}_S\in\prod_{i-1}^{s}\mathcal{K}_i$
\begin{equation}\label{4}
  \widehat{\varphi}(\widehat{g}_{In(\rho)}(\mathbf{m}_S,\mathbf{k}_S))=f(\mathbf{m}_S),
\end{equation}
and for any wiretap set $W$, the secure condition (\ref{secure}) and the user secure condition (\ref{usecure}) are satisfied respectively.
The rate of such an $(l,n)$ code $\mathbf{C}$ is defined by
\begin{equation*}
  R(\mathbf{C})=\frac{l}{n}.
\end{equation*}
We say a rate $R$ is achievable if there exists a valid network code with rate $R$ with respect to $(\mathcal{N},f,\mathcal{W})$, which implies that by using the network $n$ times, the sink node can compute target function securely (or user securely) $l$ times.
The capacity of the network $\mathcal{N}$ to compute target function $f$ with possible wiretap set $\mathcal{W}$ is the maximum achievable rate, i.e.
\begin{center}
  $\widehat{C}(\mathcal{N},f,\mathcal{W})=\max$\{$\frac{l}{n}$: there\ exists\ a\ valid\ $(l,n)$\ network\ code\}.
\end{center}

In this paper, suppose $\mathcal{A}=\mathcal{B}=\mathcal{O}=\mathbb{F}_q$, and $f$ is a linear function over $\mathbb{F}_q$, i.e.
\begin{equation*}
  f(x_1,x_2,\cdots,x_s)=\sum\limits_{i=1}^{s}a_i\cdot x_i,
\end{equation*}
where $a_i\in \mathbb{F}_q$ for all $i=1,2,\cdots,s$.
Without loss of generality, let $a_i\neq0$, $\forall i=1,2,\cdots,s$. Otherwise, we can remove the source nodes $\sigma_i$ for those $i$ such that $a_i=0$ and consider computing the same target function over this new network.
Moreover, we let $y_i=a_i\cdot x_i$, and consider the sum function $f'(y_1,y_2,\cdots,y_s)=\sum_{i=1}^{s}y_i$ over $\mathbb{F}_q$.
A secure (or user secure) network code for $(f,\mathcal{N},\mathcal{W})$ is equivalent to a secure (or user secure) network code for $(f',\mathcal{N},\mathcal{W})$.
From this observation, we only need to consider the computing capacity for the arithmetic sum function over $\mathbb{F}_q$.

Particularly, we consider the wiretap set with the form
\begin{equation*}
  \mathcal{W}_r=\{W\subset E:\vert W\vert\leq r\},
\end{equation*}
where $r$ is defined as the security level.
In the rest of this paper, we use $\widehat{C}(\mathcal{N},f,r)$ to denote $\widehat{C}(\mathcal{N},f,\mathcal{W}_r)$.
\subsection{Existing Bounds}
First, let's review some known results about secure network function computation.
For a network $\mathcal{N}$, let $C_{min}$ be the size of minimum cut of $\mathcal{N}$.
We call a cut as $\sigma_i$-cut if it is an edge subset separating $\sigma_i$ from $\rho$ and there doesn't exist any path from other sources to this cut.
For each source node $\sigma_i$, the minimum size of a $\sigma_i$-cut is denoted as $D_{min}^i$.

An important result for linear network computation over finite field is that the capacity of network $\mathcal{N}$ is $C_{min}$ \cite{AF2013}.
For secure case, Guang et al. established an upper bound \cite{GBY2021}
\begin{equation}\label{gubound}
  \widehat{C}(\mathcal{N},f,r)\leq\min\{C_{min},D_{min}-r\},
\end{equation}
where $D_{min}=\min_{\sigma_i\in S}\{D_{min}^i\}$.
On the other hand, similar to the construction of secure linear network codes, Guang et al. gave a lower bound \cite{GBY2021}
\begin{equation}\label{glbound}
  \widehat{C}(\mathcal{N},f,r)\geq C_{min}-r.
\end{equation}
Whether the upper bound is tight is not known yet.
In the scenario of user secure, Li et al. investigated the problem under the assumption of unlimited edge transmission capacity \cite{LL2022}.
A sufficient and necessary condition for the existence of user secure network codes was provided based on the network topology.

\section{An Improved Upper Bound}\label{s3}
We begin this section with an example which shows Guang's upper bound is not always tight.
\begin{figure}[htbp]
  \centering
  \includegraphics[width=0.5\textwidth]{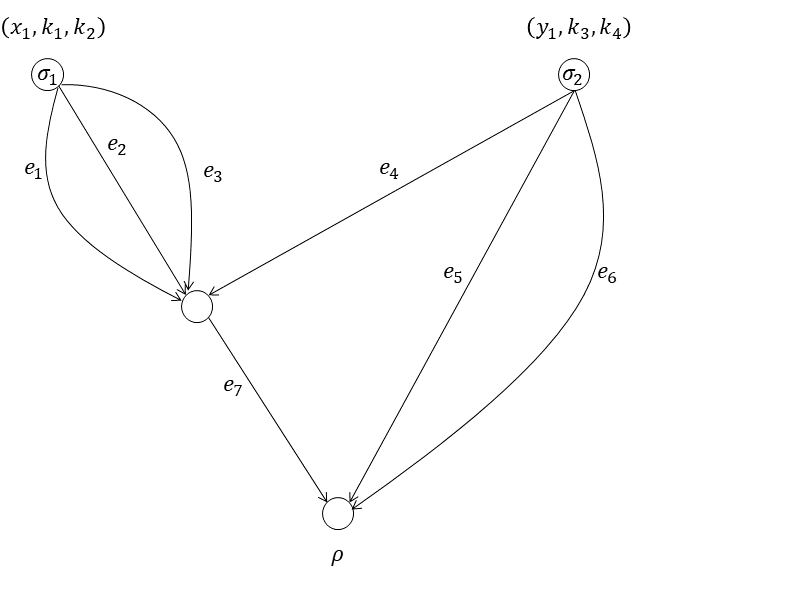}
  \caption{The network has 2 sources and the target function is $f=x_1+y_1$.}\label{F1}
\end{figure}
\begin{example}\label{e1}
  For the secure network function computation problem shown in $Fig.\ \ref{F1}$, the minimum cut and minimum $\sigma_i$-cut in this network have size 1 and 3 respectively, i.e. $C_{min}=1, D_{min}=3$.
  When the security level $r=2$, which means the wiretapper has access to any two edges, we have $\widehat{C}(\mathcal{N},f,r)\leq 1$ from (\ref{gubound}).
  In fact, there doesn't exist a secure code in this case.
  Suppose the wiretapper gets the messages transmitted by $e_4$ and $e_7$ which are denoted as $U_{e_4}$ and $U_{e_7}$.
  If there exists a secure network code, then $U_{e_4}$ is a function of the source message and the random key generated at $\sigma_2$, i.e. $U_{e_4}=g_1(\mathbf{y},\mathbf{k}')$.
  Similarly, $U_{e_7}=g_2(\mathbf{x},\mathbf{k},U_{e_4})$.
  Since the sink node can complete the computation and $U_{e_5}$, $U_{e_6}$ contain no information about $\mathbf{k}$, $U_{e_7}$ cannot contain any information about $\mathbf{k}$ either, i.e. $U_{e_7}=g_2(\mathbf{x},U_{e_4})$.
  Now, after getting $U_{e_4}$ and $U_{e_7}$ the wiretapper can get a function of $\mathbf{x}$, which contradicts the secure condition (\ref{secure}).
\end{example}

From the example, we know that Guang's upper bound is not tight although it is tight over some special $(\mathcal{N},f,r)$, see examples in \cite{GBY2021}.
For any cut $C$ in a network, let $I_C$ be the source nodes separated from $\rho$ by $C$, i.e.
\begin{center}
  $I_C=\{\sigma_i\in S:$ all\ paths\ from\ $\sigma_i$ to\ $\rho$ pass\ through\ $C\}$.
\end{center}
From a series of papers in network function computation \cite{AFK2011,HTYG2018,GYYL2019}, we know that each of those sources which can attach to $C$, i.e. there exists a path from the source to some edges in $C$, also matters to the capacity. Define
\begin{center}
  $J_C=\{\sigma_i\in S\setminus I_C:$ there\ exists\ a\ path\ from\ $\sigma_i$ to\ $C\}$.
\end{center}
Now we define a new edge set for each cut $C$ as follows. Define
\begin{center}
  $\widehat{B}=\min\{B\subset E:$ any\ path\ from\ $\sigma_i$ to $e$ pass\ through,\ $\forall\sigma_i\in J_C,\forall e\in C\}$,
\end{center}
where the minimum is taken with respect to the size of $B$.
For any cut $C$, let $A=C\cup\widehat{B}$, and define $A_{min}=\min_{C\in\Lambda}\{|A|:A=C\cup\widehat{B}\}$, where $\Lambda$ is the collection of all cut sets of the network.

We present our upper bound in the following theorem.
\begin{theorem}\label{th1}
  Let $\mathcal{N}$ be a network and $f$ be an algebraic sum over a finite field $\mathbb{F}_q$.
  If the security level $r\leq A_{min}$, then
  \begin{equation}\label{nubound}
    \widehat{C}(\mathcal{N},f,r)\leq\min\{C_{min},A_{min}-r\},
  \end{equation}
  and $\widehat{C}(\mathcal{N},f,r)=0$ otherwise.
\end{theorem}
\begin{proof}
  For the secure network function computation problem $(\mathcal{N},f,r)$, suppose there exists an $(l,n)$ secure code.
  To prove the theorem, we only need to show $l/n\leq A_{min}-r$.

  Let the message generated at source node $\sigma_i$ be a random variable $\mathbf{X}_i=(\mathbf{M}_i,\mathbf{K}_i)=(M_{i1},\cdots,M_{il},K_{i1},\cdots,K_{ir})$, where $M_{ij},K_{ij'}, i\in[s],\ j\in[l],\ j'\in[r]$ are generated independently and uniformly at random over $\mathbb{F}_q$.
  Then by the definition of entropy function, we have
  \begin{equation*}
    \mathbf{H}(\mathbf{M}_i)=l\cdot\mathbf{H}(M_{ij})=l\cdot\log q,
  \end{equation*}
  where the logarithm base is two. The message transmitted by any edge $e$ is denoted as $U_{e}\in\mathbb{F}_q^n$.

  Let $A=C\cup B$ be the edge set with size $A_{min}$, where $C$ is a cut of the network and B is the minimum cut separating $J_C$ and $C$.
  Without loss of generality, we assume $I_C=\{\sigma_1\}$, and consider a wiretap set $W\subset A$ with size $r$.
  Let $\mathbf{Y}_W=\{U_e:e\in W\}$, and $\mathbf{Y}_A=\{U_e:e\in A\}$.
  Since each source generates message independently, we have
  \begin{equation}\label{p1}
    l\cdot\log q=\mathbf{H}(\mathbf{M}_1)=\mathbf{H}(\mathbf{M}_1|\mathbf{M}_{S\setminus D}),
  \end{equation}
  where $D=J_C\cup\{\sigma_1\}$. The secure condition (\ref{secure}) implies
  \begin{center}
    $\mathbf{H}(\mathbf{M}_1,\mathbf{M}_{S\setminus D}|\mathbf{Y}_W)=\mathbf{H}(\mathbf{M}_1,\mathbf{M}_{S\setminus D}),$\\
    $\mathbf{H}(\mathbf{M}_{S\setminus D}|\mathbf{Y}_W)=\mathbf{H}(\mathbf{M}_{S\setminus D})$.
  \end{center}
  Expanding both sides of the first equation, we get
  \begin{align*}
    \mathbf{H}(\mathbf{M}_1|\mathbf{M}_{S\setminus D},\mathbf{Y}_W)+ \mathbf{H}(\mathbf{M}_{S\setminus D}|\mathbf{Y}_W)
    & =\mathbf{H}(\mathbf{M}_1,\mathbf{M}_{S\setminus D}|\mathbf{Y}_W)\\
    & =\mathbf{H}(\mathbf{M}_1,\mathbf{M}_{S\setminus D})\\
    & =\mathbf{H}(\mathbf{M}_1|\mathbf{M}_{S\setminus D})+\mathbf{H}(\mathbf{M}_{S\setminus D}),
  \end{align*}
  which shows
  \begin{equation}\label{p2}
    \mathbf{H}(\mathbf{M}_1|\mathbf{M}_{S\setminus D},\mathbf{Y}_W) = \mathbf{H}(\mathbf{M}_1|\mathbf{M}_{S\setminus D}).
  \end{equation}
  By the definition of edge set $A$, we know that every edge in $A$ transmits a function of messages generated at sources in $D=J_C\cup\{\sigma_1\}$, which means $\mathbf{Y}_W$ is a function of $\mathbf{M}_D$ and $\mathbf{K}_D$. Then,
  \begin{equation}\label{p3}
    \mathbf{H}(\mathbf{M}_1|\mathbf{M}_{S\setminus D},\mathbf{Y}_W) = \mathbf{H}(\mathbf{M}_1|\mathbf{M}_{S\setminus D},\mathbf{Y}_W,\mathbf{K}_{S\setminus D}).
  \end{equation}
  We claim that
  \begin{equation}\label{p4}
    \mathbf{H}(\mathbf{M}_1|\mathbf{M}_{S\setminus D},\mathbf{Y}_A,\mathbf{K}_{S\setminus D})=0.
  \end{equation}
  This claim can be shown as follows
  \begin{align*}
    \mathbf{H}(\mathbf{M}_1|\mathbf{M}_{S\setminus D},\mathbf{Y}_A,\mathbf{K}_{S\setminus D})
    & = \mathbf{H}(\mathbf{M}_1|\mathbf{M}_{S\setminus \{1\}},\mathbf{Y}_A,\mathbf{K}_{S\setminus \{1\}})\\
    & = \mathbf{H}(\mathbf{M}_1|\mathbf{M}_{S\setminus \{1\}},\mathbf{Y}_A,\mathbf{K}_{S\setminus \{1\}},\mathbf{Y}_{C'})\\
    & = \mathbf{H}(\mathbf{M}_1|\mathbf{M}_{S\setminus \{1\}},\mathbf{K}_{S\setminus \{1\}},f(\mathbf{M}_1,\cdots,\mathbf{M}_s))\\
    & \leq \mathbf{H}(\mathbf{M}_1|\mathbf{M}_{S\setminus \{1\}},f(\mathbf{M}_1,\cdots,\mathbf{M}_s))=0,
  \end{align*}
  where $C'=\bigcup_{i\in S\setminus \{1\}}Out(\sigma_i)$.
  The first equality holds for the independence of source messages and keys.
  By the definition of $C'$, $\mathbf{Y}_{C'}$ is a function of $\mathbf{M}_{S\setminus \{1\}}$ and $\mathbf{K}_{S\setminus \{1\}}$, then the second equality follows.
  The third equation holds since that $A\cup C'$ is a global cut of the network, the target function can be computed for any global cut messages.
  And the last equation is from that $f(\mathbf{M}_1,\cdots,\mathbf{M}_s)=\mathbf{M}_1+\cdots+\mathbf{M}_s$ and $\mathbf{M}_{S\setminus \{1\}}$ is known.

  Combining (\ref{p1}), (\ref{p2}), (\ref{p3}), (\ref{p4}), we get
  \begin{align*}
    l\cdot\log q
    & = \mathbf{H}(\mathbf{M}_1|\mathbf{M}_{S\setminus D},\mathbf{Y}_W,\mathbf{K}_{S\setminus D})\\
    & = \mathbf{H}(\mathbf{M}_1|\mathbf{M}_{S\setminus D},\mathbf{Y}_W,\mathbf{K}_{S\setminus D})-
      \mathbf{H}(\mathbf{M}_1|\mathbf{M}_{S\setminus D},\mathbf{Y}_A,\mathbf{K}_{S\setminus D})\\
    & = \mathbf{I}(\mathbf{M}_1;\mathbf{Y}_{A\setminus W}|\mathbf{M}_{S\setminus D},\mathbf{K}_{S\setminus D})\\
    & \leq  \mathbf{H}(\mathbf{Y}_{A\setminus W}) \leq (|A|-r)\cdot n\cdot\log q.
  \end{align*}
  Then we can conclude that
  \begin{equation*}
    \frac{l}{n}\leq |A|-r = A_{min}-r.
  \end{equation*}
\end{proof}
\begin{remark}
  Comparing with Guang's upper bound (\ref{gubound}), we only change $D_{min}$ to $A_{min}$.
  By the definition of $A_{min}$, if $C$ is a $\sigma_i$-cut of the network, then $B=\emptyset$ since $J_C=\emptyset$, then $A=C\cup B=C$.
  Thus $D_{min}$ is the minimum size of $A$, where the minimum is taken over all $\sigma_i$-cut, and $A_{min}$ is the minimum taken over all cut.
  Therefore, $A_{min}\leq D_{min}$ in general, which implies our new upper bound is tighter than the previous one.
\end{remark}
\begin{example}
  Consider the network function computation problem shown in $Fig.\ \ref{F1}$.
  For $r=1$, our new upper bound shows $\widehat{C}(\mathcal{N},f,1)\leq 1$, which is achievable and the secure code is the same as the example in \cite{GBY2021}.
  For $r=2$, applying Theorem \ref{th1}, we have $\widehat{C}(\mathcal{N},f,2)\leq 0$, since $A_{min}=2$ for the network.
  This also matches our analysis in Example \ref{e1}.
\end{example}
\section{Lower Bound}\label{s4}
For the lower bound of the secure computing capacity over a network, we need to construct an algorithm to generate a secure $(l,n)$-code for any given network, target function and security level, which is difficult in general.
For non-secure case, most known results about lower bound are obtained by connecting the problem with equivalent network coding problem \cite{RD2012,KEH2004}.
Similar to the construction of secure network coding problem, Guang et al. \cite{GBY2021} provided a lower bound with rate $C_{min}-r$ by a linear transformation at the source nodes.
In the following, we address a sufficient and necessary condition of the linear transformation which turns a non-secure network code into a secure one, and prove that for multi-edge tree graph, the upper bound in Theorem \ref{th1} is tight.

Recall that an $(l,n)$ network code can be represented by global encoding functions for all edges in the network, i.e. $\{\widehat{g}_e:e\in E\}$.

We denote the message generated at each source node $\sigma_i$ as $\mathbf{M}_i\in\mathbb{F}_q^{n-r}$, and denote the random key as $\mathbf{K}_i\in\mathbb{F}_q^{r}$ in this section.
Each edge $e\in E$ transmits a message of length 1 over $\mathbb{F}_q$.
Such a linear secure code has rate $n-r$, and each global encoding function can be represented by a column vector $g_e\in \mathbb{F}_q^{sn}$, i.e.
\begin{equation*}
  \widehat{g}_e(\mathbf{M}_1,\mathbf{K}_1,\cdots,\mathbf{M}_s,\mathbf{K}_s)=(\mathbf{M}_1,\mathbf{K}_1,\cdots,\mathbf{M}_s,\mathbf{K}_s)\cdot g_e.
\end{equation*}
A linear transformation at all source nodes is an $sn\times sn$ matrix $A$ which is a block diagonal matrix,
\begin{equation*}
  A=\begin{pmatrix}
      A_1 & 0 & \cdots & 0 \\
      0 & A_2 & \cdots & 0 \\
      \vdots & \vdots & \ddots & \vdots \\
      0 & 0 & \cdots & A_s
    \end{pmatrix},
\end{equation*}
where $A_i$ is an $n\times n$ transformation matrix at source node $\sigma_i$.
For $T\subset[sn]$, let $A_T$ be the matrix obtained by selecting the columns of $A$ corresponding to $T$.
For instance, $A_{\{1\}}$ is the first column of $A$.
Then, we can give a sufficient and necessary condition for the linear transformation transforming a non-secure code into a secure one.

\begin{theorem}\label{th2}
  Given the network $\mathcal{N}$, sum function $f$, and security level $r$, if there exists a non-secure network code $\{\widehat{g}_e:e\in E\}$ with rate $R=n$, then it can be transformed into a secure network code with rate $n-r$ by a linear transformation $A$ if and only if $A$ satisfies
  \begin{equation}\label{condition}
    span(A_T^{-1})\cap span(F(W))=\mathbf{0},
  \end{equation}
  where $T=\bigcup_{i=0}^{s-1}\{i+1,i+2,\cdots,i+n-r\}$.
\end{theorem}
\begin{proof}
  Denote all source messages and random keys as a row vector of length $sn$ by $\mathbf{X}=(\mathbf{M}_1,\mathbf{K}_1,\cdots,\mathbf{M}_s,\mathbf{K}_s)$.
  By the definition of global encoding functions, before transformation, edge $e\in E$ sends a message $U_e=\mathbf{X}\cdot g_e$.
  After transformation, the message sent by edge $e$ is $U_e=\mathbf{X}\cdot A \cdot g_e$.
  Let $W\subset E$ be a wiretap edge subset, $G(W)=\{g_e:e\in W\}$, the wiretapper can get
  \begin{equation*}
    \mathbf{Y}_W=\mathbf{X}\cdot A \cdot G(W).
  \end{equation*}
  Here, the security condition
  \begin{equation*}
    \mathbf{I}(\mathbf{M}_1,\mathbf{M}_2,\cdots,\mathbf{M}_s;\mathbf{Y}_W)=0
  \end{equation*}
  is equivalent to that there doesn't exist any vector $\mathbf{a}\in\mathbb{F}_q^{s(n-r)},\mathbf{b}\in\mathbb{F}_q^{|W|}$ such that
  \begin{equation}\label{dj1}
    (\mathbf{M}_1,\mathbf{M}_2,\cdots,\mathbf{M}_s)\cdot \mathbf{a}= \mathbf{Y}_W\cdot \mathbf{b},
  \end{equation}
  since
  \begin{equation*}
    (\mathbf{M}_1,\mathbf{M}_2,\cdots,\mathbf{M}_s)=\mathbf{X}\cdot\begin{pmatrix}
                                                                     J_{n,n-r} & 0 & \cdots & 0 \\
                                                                     0 & J_{n,n-r} & \cdots & 0 \\
                                                                     \vdots & \vdots & \ddots & \vdots \\
                                                                     0 & 0 & \cdots & J_{n,n-r}
                                                                   \end{pmatrix},
  \end{equation*}
  where $J_{n,n-r}$ is the first $n-r$ columns of an identity matrix of order $n$.
  From above, Eq. (\ref{dj1}) is equivalent to
  \begin{equation}\label{dj2}
    \mathbf{X}\cdot\begin{pmatrix}
    J_{n,n-r} & 0 & \cdots & 0 \\
    0 & J_{n,n-r} & \cdots & 0 \\
    \vdots & \vdots & \ddots & \vdots \\
    0 & 0 & \cdots & J_{n,n-r}
    \end{pmatrix}\cdot \mathbf{a}= \mathbf{X}\cdot A \cdot G(W)\cdot \mathbf{b},
  \end{equation}
  then we have
  \begin{equation}\label{dj3}
    A^{-1}\cdot\begin{pmatrix}
    J_{n,n-r} & 0 & \cdots & 0 \\
    0 & J_{n,n-r} & \cdots & 0 \\
    \vdots & \vdots & \ddots & \vdots \\
    0 & 0 & \cdots & J_{n,n-r}
    \end{pmatrix}\cdot \mathbf{a}= G(W)\cdot \mathbf{b},
  \end{equation}
  the lefthand side of Eq. (\ref{dj3}) can be represented by $A_T^{-1}$, where $T=\bigcup_{i=0}^{s-1}\{i+1,i+2,\cdots,i+n-r\}$.
  This is equivalent to $span(A_T^{-1})\cap span(F(W))=\mathbf{0}$.
\end{proof}

For sum-network, the capacity of network $\mathcal{N}$ is equal to $C_{min}$ \cite{RD2012,KEH2004}.
Thus from Theorem \ref{th2}, for secure network sum function computing problem, we have $\widehat{C}(\mathcal{N},f,r)\geq C_{min}-r$.

We call a graph $\mathcal{N}=(V,E)$ as multi-edge tree if for any node $v\in V$, there exists a node $u$ such that all the out edges of $v$ are in edges of $u$, which means $Out(v)\subset In(u)$. We obtain the following corollary.

\begin{corollary}
  For multi-edge tree graphs, the upper bound in Theorem \ref{th1} is tight.
\end{corollary}
\begin{proof}
  In multi-edge tree graphs, for any cut $C$, $J_{C}=\emptyset$. Thus, $A_{min}=C_{min}$, and by Theorem \ref{th1}, the capacity is $C_{min}-r$.
  And in Theorem \ref{th2}, there exists a secure network code with rate $C_{min}-r$.
\end{proof}

For general network which satisfies $A_{min}>C_{min}$, whether the upper bound (\ref{nubound}) is tight or not is not known yet. But we can provide some necessary conditions for linear secure global encoding functions.
Let $\widehat{g}_e(\mathbf{M}_1,\mathbf{K}_1,\cdots,\mathbf{M}_s,\mathbf{K}_s)=(\mathbf{M}_1,\mathbf{K}_1,\cdots,\mathbf{M}_s,\mathbf{K}_s)\cdot g_e$, where $g_e$ is a column vector with length $sn$.
For any source node $\sigma_i$, let $\widehat{g}_e|_{\{\sigma_i\}}$ be the global encoding function restricted on $\mathbf{M}_i,\mathbf{K}_i$, i.e.
\begin{equation*}
  \widehat{g}_e|_{{\sigma_i}}(\mathbf{M}_i,\mathbf{K}_i)=(\mathbf{M}_i,\mathbf{K}_i)\cdot g_e|_{\{\sigma_i\}}.
\end{equation*}
\begin{theorem}
  If there exists an $(n-r,1)$ linear secure network code $\{\widehat{g}_e:e\in E\}$ for sum-network problem $(\mathcal{N},f,r)$, then for any cut $C$, the messages $\{\widehat{g}_e|_{I_C}:e\in C\}$ contain all information about $\sum_{\sigma_i\in I_C}\mathbf{M}_i$.
\end{theorem}
\begin{proof}
  Suppose we cannot decode $\sum_{\sigma_i\in I_C}\mathbf{M}_i$ from the truncated messages $\{\widehat{g}_e:e\in C\}$.
  Then there must be some random keys from source nodes $\sigma_i\in I_C$ we cannot eliminate.
  Let $G=C\cup(\cup_{\sigma_i\notin I_C}Out(\sigma_i))$ be a global cut, we can decode the sum from the messages transmitted on these edges.
  However, $\cup_{\sigma_i\notin I_C}Out(\sigma_i)$ don't contain any message about $\mathbf{K}_i$, for $\sigma_i\in I_C$, which means the random keys in $\{\widehat{g}_e:e\in C\}$ cannot be eliminated. Then it's impossible to obtain the sum from the global cut.
\end{proof}

For the case $A_{min}\geq C_{min}$, if the wiretapper has access to a cut $C$ of the network with size less than $A_{min}$, he cannot get any information about any source message from the secure condition. However, from above theorem, the sum $\sum_{\sigma_i\in I_C}\mathbf{M}_i$ can be obtained from truncated messages $\{\widehat{g}_e|_{I_C}\:e\in C\}$, which indicates that the messages in the min-cut from $J_C$ to $C$ are served as random keys to protect the sum from the wiretapper.

\section{Upper Bound under User Secure}\label{s5}
In this section, we consider the network function computation problem under the user secure condition (\ref{usecure}).
Li et al. referred this kind of security as secrete sum when the target function is arithmetic sum over finite field \cite{LL2022}.
In their model, each edge of the network has unlimited transmission capacity, which means each edge can transmit a message over $\mathbb{F}_q$ with unlimited length once a time.
Therefore, it is meaningless to bound the capacity under their model.
In \cite{LL2022}, the sufficient and necessary condition for a user secure network function computation problem is that after deleting all wiretap edges, there exists at least one source having a path to the sink node $\rho$.

Now, we assume the transmission capacity of each edge is unit, i.e. each edge transmits one symbol over $\mathbb{F}_q$ at one time.
What is the relation between network topology and computation capacity under a given security level $r$?
Here, we will use $\widehat{C}_{us}(\mathcal{N},f,r)$ to denote the capacity of the network $\mathcal{N}$ computing $f$ with security level $r$.
We first address a sufficient and necessary condition for the existence of user secure codes and then provide an upper bound on the capacity based on network topology.
We begin with the following illustrative examples.
\begin{figure}[htbp]
  \centering
  \includegraphics[width=0.4\textwidth]{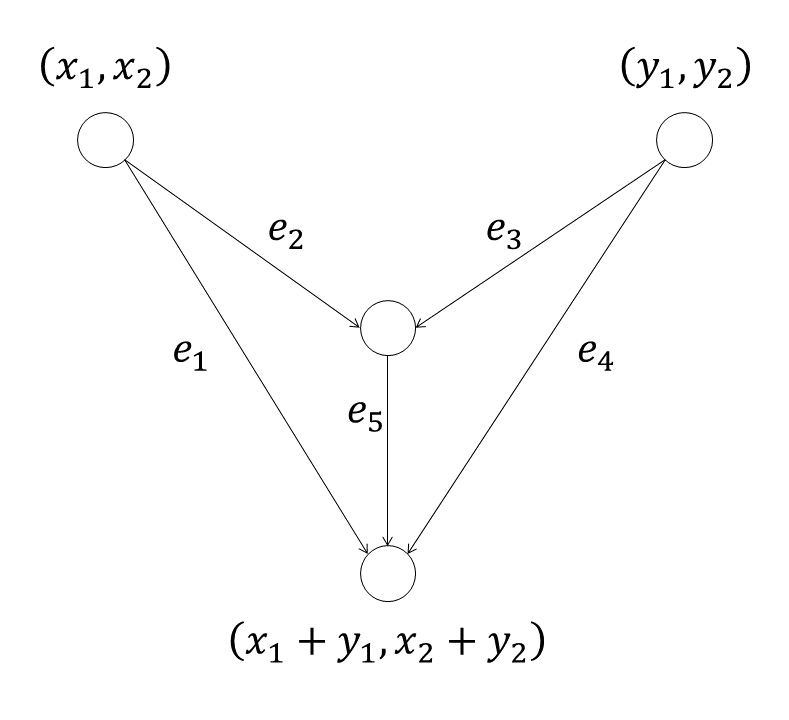}
  \caption{The network has 2 sources and the target function is $f=(x_1+y_1,x_2+y_2)$.}\label{F2}
\end{figure}

\begin{example}\label{e2}
  Consider the user secure network function computation problem shown in $Fig.\ \ref{F2}$.
  When the security level $r=1$, i.e. the wiretapper has access to exactly one edge of the network, there exists a (2,1) network code over $\mathbb{F}_2$ defined as follows:
  $U_{e_1}=x_1+x_2, U_{e_2}=x_2, U_{e_3}=y_1, U_{e_4}=y_1+y_2,$ and $U_{e_5}=x_2+y_1$.
  The sink node $\rho$ can decode $x_1+y_1$ and $x_2+y_2$ from $U_{e_1},U_{e_4},U_{e_5}$, and any one of $U_{e_1},\cdots, U_{e_5}$ cannot leakage any information about $(x_1+y_1,x_2+y_2)$.
  This example shows that $\widehat{C}_{us}(\mathcal{N},f,r)\geq 2$.
  On the other hand, for sum function computation task over network, the capacity cannot be greater than $C_{min}$.
  Otherwise, we will have a code with rate larger than $C_{min}$ which contradicts with the result in \cite{RD2012}.
  Thus, in this example, $\widehat{C}_{us}(\mathcal{N},f,r)=2$ is tight.
\end{example}

The fundamental reason why random keys are not needed is that the user secure condition (\ref{usecure}) means the leakage messages $\mathbf{Y}_W$ should not be a function of target result $\mathbf{X}_1+\cdots+\mathbf{X}_s$.
In above example, even the wiretapper knows $x_2$, he cannot get any information about $x_2+y_2$ since $y_2$ is uniformly distributed on $\mathbb{F}_q$.
Therefore, we focus on the capacity without random keys. In fact, if we allow random keys, we may have larger capacity since if a code is secure then it is also user secure.
\begin{example}\label{e4}
  In Example \ref{e2}, we claim that when the security level $r=2$, there doesn't exist any user secure network code such that $\rho$ can compute $\mathbf{X}+\mathbf{Y}$ without error.
  Suppose there exists an $(l,n)$ user secure code, i.e. source messages $\mathbf{X},\mathbf{Y}\in\mathbb{F}_q^l$, edge $e$ transmits a message $U_e\in\mathbb{F}_q^n$, and from solvability and user security, we have
  \begin{align*}
    \mathbf{H}(\mathbf{X}+\mathbf{Y}|U_{e_1},U_{e_4},U_{e_5})&=0,\\
    \mathbf{I}(\mathbf{X}+\mathbf{Y};U_{e_1},U_{e_4})&=0,\\
    \mathbf{I}(\mathbf{X}+\mathbf{Y};U_{e_1},U_{e_5})&=0,
  \end{align*}
  where $U_{e_1}$ is a function of $\mathbf{X}$, and $U_{e_4}$ is a function of $\mathbf{Y}$. Then we will have
  \begin{align*}
    \mathbf{I}(U_{e_1},U_{e_5};\mathbf{X}+\mathbf{Y}) & \geq \mathbf{I}(U_{e_1},U_{e_5};\mathbf{X}+\mathbf{Y}|\mathbf{Y}) \\
     & = \mathbf{H}(\mathbf{X}+\mathbf{Y}|\mathbf{Y})- \mathbf{H}(\mathbf{X}+\mathbf{Y}|U_{e_1},U_{e_5},\mathbf{Y}) \\
     & = \mathbf{H}(\mathbf{X})- \mathbf{H}(\mathbf{X}+\mathbf{Y}|U_{e_1},U_{e_5},U_{e_4},\mathbf{Y})\\
     & = \mathbf{H}(\mathbf{X})-0>0,
  \end{align*}
  which contradicts the user security.
  The first equation is from the definition of mutual information, and the second equation holds since $\mathbf{X}$ and $\mathbf{Y}$ are independent and uniformly distributed over $\mathbb{F}_q^l$ and $U_{e_4}$ is a function of $\mathbf{Y}$.
  From another view, consider $\mathbf{I}(\mathbf{X}+\mathbf{Y};U_{e_1},U_{e_4})$, we can prove
  \begin{align*}
    \mathbf{I}(\mathbf{X}+\mathbf{Y};U_{e_1},U_{e_4}) & \geq \mathbf{I}(\mathbf{X}+\mathbf{Y};U_{e_1},U_{e_4}|\mathbf{Y}) \\
     & = \mathbf{H}(U_{e_1},U_{e_4}|\mathbf{Y})+\mathbf{H}(\mathbf{X}+\mathbf{Y}|\mathbf{Y})-\mathbf{H}(\mathbf{X}+\mathbf{Y},U_{e_1},U_{e_4}|\mathbf{Y}) \\
     & = \mathbf{H}(U_{e_1})+\mathbf{H}(\mathbf{X})-\mathbf{H}(\mathbf{X}+\mathbf{Y},U_{e_1}|\mathbf{Y}) \\
     & = \mathbf{H}(U_{e_1})+\mathbf{H}(\mathbf{X})-\mathbf{H}(\mathbf{X},U_{e_1})\\
     & = \mathbf{I}(\mathbf{X};U_{e_1}) = \mathbf{H}(U_{e_1}) >0,
  \end{align*}
  which also contradicts with user security. The second equation holds since $U_{e_4}$ is determined by $\mathbf{Y}$.
\end{example}

From the above examples, we can conclude a sufficient and necessary condition for the existence of user secure network codes.

\begin{theorem}\label{th43}
  Let $\mathcal{N}$ be a network and $f$ be a sum function over $\mathbb{F}_q$. Suppose the security level is $r$, then there exists a user secure network code if and only if $r<\min\{C_{min}, s\}$, where $s$ is the number of source nodes.
\end{theorem}
\begin{proof}
  We first prove the only if part.
  Denote the source messages as $\mathbf{X}_1,\cdots,\mathbf{X}_s$, the sink node $\rho$ wants $\mathbf{X}_1+\cdots+\mathbf{X}_s$, while the wiretapper has access to any $r$ edges and cannot get any information about the sum.
  We prove this part in two steps.

  First, suppose $r\geq C_{min}$ and let $C$ be a cut with $|C|=C_{min}$.
  Without loss of generality, assume $I_C=\{\sigma_1\}$, let $C'= \cup_{i=2}^sOut(\sigma_i)$.
  Note that $G=C\cup C'$ is a global cut of the network, thus the sum can be decoded from messages transmitted by $G$, i.e.
  \begin{equation*}
     \mathbf{H}(\mathbf{X}_1+\cdots+\mathbf{X}_s|\{U_e,e\in G\})=0.
  \end{equation*}
  If there exists a user secure code with security level $r\geq C_{min}$, then the sum is protected from the wiretapper who has all messages in cut $C$, which indicates
  \begin{equation*}
    I(\mathbf{X}_1+\cdots+\mathbf{X}_s;\{U_e,e\in C\})=0.
  \end{equation*}
  Similar with the process in Example \ref{e4},
  \begin{align*}
    \mathbf{I}(\{U_e,e\in C\};\mathbf{X}_1+\cdots+\mathbf{X}_s)
     & \geq \mathbf{I}(\{U_e,e\in C\};\mathbf{X}_1+\cdots+\mathbf{X}_s|\{\mathbf{X}_2,\cdots,\mathbf{X}_s\}) \\
     & = \mathbf{H}(\mathbf{X}_1+\cdots+\mathbf{X}_s|\{\mathbf{X}_2,\cdots,\mathbf{X}_s\})-
         \mathbf{H}(\mathbf{X}_1+\cdots+\mathbf{X}_s|\{U_e,e\in C\},\{\mathbf{X}_2,\cdots,\mathbf{X}_s\}) \\
     & = \mathbf{H}(\mathbf{X}_1)- \mathbf{H}(\mathbf{X}_1+\cdots+\mathbf{X}_s|\{U_e,e\in C\},\{U_e,e\in C'\},\{\mathbf{X}_2,\cdots,\mathbf{X}_s\})\\
     & = \mathbf{H}(\mathbf{X}_1)-0>0,
  \end{align*}
  which contradicts with the secure condition above.
  The third equation holds since $\{U_e,e\in C'\}$ is determined by $\{\mathbf{X}_2,\cdots,\mathbf{X}_s\}$. And the last equation follows from that any global cut must carry all messages about $\mathbf{X}_1+\cdots+\mathbf{X}_s$.

  Similarly, suppose there is a user secure code with security level $r\geq s$, then for the edge set $E'=\{e_i:i\in[s],e_i\in Out(\sigma_i)\}$ with size $s$, the follow equation must be satisfied,
  \begin{equation*}
     \mathbf{I}(\mathbf{X}_1+\cdots+\mathbf{X}_s;\{U_{e_i},e_i\in E'\})=0.
  \end{equation*}
  Consider the conditional mutual information,
  \begin{align*}
    \mathbf{I}(\{U_{e_i},e_i\in E'\};\sum\limits_{i=1}^{s}\mathbf{X}_i)
    & \geq \mathbf{I}(\{U_{e_i},e_i\in E'\};\sum\limits_{i=1}^{s}\mathbf{X}_i|\{\mathbf{X}_2,\cdots,\mathbf{X}_s\}) \\
    & = \mathbf{H}(\sum\limits_{i=1}^{s}\mathbf{X}_i|\{\mathbf{X}_2,\cdots,\mathbf{X}_s\})
        +\mathbf{H}(\{U_{e_i},e_i\in E'\}|\{\mathbf{X}_2,\cdots,\mathbf{X}_s\})\\
       &\ \ \  -\mathbf{H}(\sum\limits_{i=1}^{s}\mathbf{X}_i,\{U_{e_i},e_i\in E'\}|\{\mathbf{X}_2,\cdots,\mathbf{X}_s\})\\
    & = \mathbf{H}(\mathbf{X}_1)+\mathbf{H}(U_{e_1})-\mathbf{H}(\mathbf{X}_1,U_{e_1})\\
    & = \mathbf{I}(\mathbf{X};U_{e_1}) = \mathbf{H}(U_{e_1}) >0,
  \end{align*}
  which is also a contradiction. Therefore, when $r\geq s$ or $r\geq C_{min}$, there doesn't exist any user secure code.

  For the if part, we claim that if $r<\min\{C_{min},s\}$, then there exists a user secure code with rate 1.
  The construction is as follows.
  For each source node $\sigma_i$, we select one path from $\sigma_i$ to $\rho$.
  Each source node sends its message $x_i$ directly and other nodes send the sum of messages they get.
  This network code has rate 1 and is user secure since $r<C_{min}$ and $r<s$, there exists at least one source and corresponding message cannot be observed by the wiretapper.
\end{proof}

From Theorem \ref{th43}, we know that there always exists a user secure network code when $r<\min\{C_{min},s\}$.
The following example shows $C_{min}$ is not always achievable, which means the study of user secure capacity is not trivial.
\begin{figure}[htbp]
  \centering
  \includegraphics[width=0.4\textwidth]{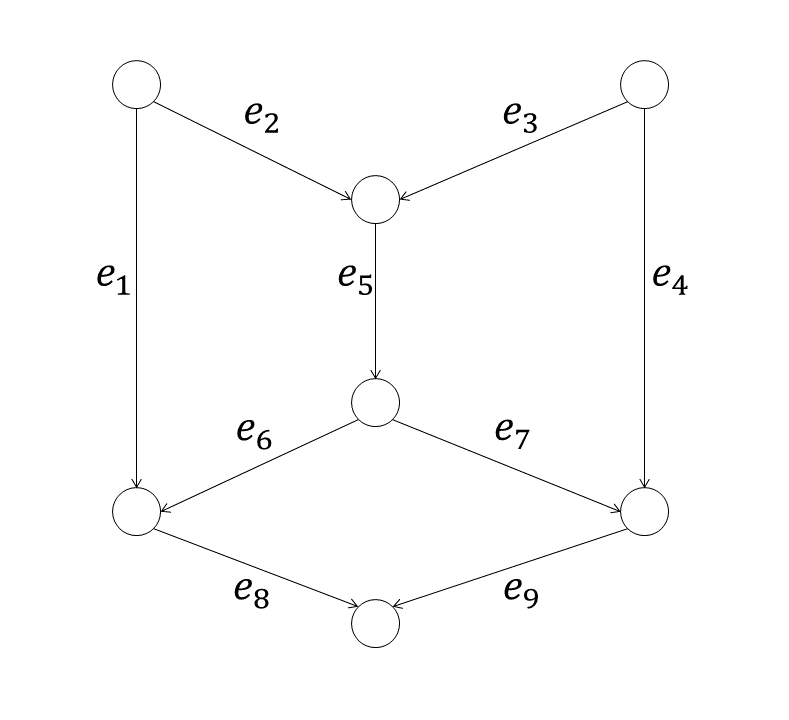}
  \caption{The network has 2 sources and the target function is $f=x+y$}\label{F3}
\end{figure}

\begin{example}\label{e3}
  Consider the reverse butterfly network shown in $Fig. \ref{F3}$.
  We claim that there doesn't exist a rate 2 network code for security level $r=1$.
  Suppose there exists a user secure network code with rate 2.
  Consider the global cut $\{e_8, e_9\}$, and let the wiretap edge $W=\{e_8\}$.
  By user secure condition (\ref{usecure}),
  \begin{equation*}
    \mathbf{H}(\mathbf{M}_1+\mathbf{M}_2|\mathbf{Y}_{e_8})=\mathbf{H}(\mathbf{M}_1+\mathbf{M}_2).
  \end{equation*}
  By the solvability of the code, we have
  \begin{equation*}
    \mathbf{H}(\mathbf{M}_1+\mathbf{M}_2|\mathbf{Y}_{e_8},\mathbf{Y}_{e_9})=0.
  \end{equation*}
  Combining the above equations,
  \begin{equation*}
     \mathbf{I}(\mathbf{M}_1+\mathbf{M}_2;\mathbf{Y}_{e_9}|\mathbf{Y}_{e_8})=\mathbf{H}(\mathbf{M}_1+\mathbf{M}_2|\mathbf{Y}_{e_8})
     -\mathbf{H}(\mathbf{M}_1+\mathbf{M}_2|\mathbf{Y}_{e_8},\mathbf{Y}_{e_9})=\mathbf{H}(\mathbf{M}_1+\mathbf{M}_2)>0,
  \end{equation*}
  which means $\mathbf{I}(\mathbf{M}_1+\mathbf{M}_2;\mathbf{Y}_{e_9})>0$. This means if wiretap edge $W=\{e_9\}$, the wiretapper will get some information about $\mathbf{M}_1+\mathbf{M}_2$, which is a contradiction.
\end{example}

Let $G$ be a global cut of the network, and $G_{min}$ be the smallest size of a global cut.
From Example \ref{e3} we can conclude the following theorem.
\begin{theorem}
  Let $\mathcal{N}$ be a network and $f$ be an algebraic sum over a finite field $\mathbb{F}_q$.
  If the security level $r<\min\{C_{min}, s\}$, then
  \begin{equation}\label{ubound}
    \widehat{C}_{us}(\mathcal{N},f,r)\leq\min\{C_{min},G_{min}-r\},
  \end{equation}
  and $\widehat{C}_{us}(\mathcal{N},f,r)=0$ otherwise.
\end{theorem}
\begin{proof}
  We only prove $\widehat{C}_{us}(\mathcal{N},f,r)\leq G_{min}-r$, since $\widehat{C}_{us}(\mathcal{N},f,r)\leq C_{min}$ can be obtained from \cite{RD2012}.
  In fact, there doesn't exist a network code which can compute arithmetic sum on a network with rate larger than $C_{min}$ \cite{RD2012}.
  Similar to Example \ref{e3}, assume $G$ is a global cut with size $G_{min}$, and $W\subset G$ is a wiretap set.
  Suppose there exists an $(l,n)$ user secure code.
  Since source messages are uniformly and independently distributed,
  \begin{equation*}
    \mathbf{H}(\mathbf{M}_1+\mathbf{M}_2+\cdots+\mathbf{M}_s)=l\cdot\log q.
  \end{equation*}
   By user secure condition (\ref{usecure}),
  \begin{equation*}
    \mathbf{H}(\mathbf{M}_1+\cdots+\mathbf{M}_s|\mathbf{Y}_W)=\mathbf{H}(\mathbf{M}_1+\cdots+\mathbf{M}_s).
  \end{equation*}
  By the solvability of the code, we have
  \begin{equation*}
    \mathbf{H}(\mathbf{M}_1\cdots+\mathbf{M}_s|\mathbf{Y}_G)=0.
  \end{equation*}
  Combining the above equations,
  \begin{align*}
     (|G|-r)\cdot n\cdot\log q&\geq \mathbf{H}(\mathbf{Y}_{G\setminus W})\\
     & \geq \mathbf{I}(\mathbf{M}_1+\mathbf{M}_2;\mathbf{Y}_{G\setminus W}|\mathbf{Y}_{W})\\
     & = \mathbf{H}(\mathbf{M}_1+\mathbf{M}_2|\mathbf{Y}_{W})-\mathbf{H}(\mathbf{M}_1+\mathbf{M}_2|\mathbf{Y}_G)\\
     & = \mathbf{H}(\mathbf{M}_1+\mathbf{M}_2)=l\cdot\log q,
  \end{align*}
  which shows $l/n\leq G_{min}-r$.
\end{proof}

\section{Conclusions}\label{s6}
In this paper, we first present a new upper bound of the network computation capacity under the secure condition.
The new upper bound is strictly tighter than the previous one. When the network is a multi-edge tree, we can prove that the new upper bound is tight.
Meanwhile, we give a necessary condition for the existence of linear secure network codes for computing a sum function over network, which provides some insight of the construction of achievable network codes.

On the other hand, we consider the capacity problem under another new secure condition defined as user secure in this paper.
We first give a sufficient and necessary condition for the existence of user secure network codes.
Furthermore, an upper bound of the user secure computation capacity is addressed.
However, whether the bound is tight or not remains unknown yet.

\end{document}